\documentclass{amsart}

\usepackage{epsfig,amsfonts,amssymb,amsmath,amsthm,eucal,pinlabel,array,stmaryrd}
\usepackage{calrsfs,bbm,enumerate}

\newtheorem{theorem}{Theorem}[section]

\newtheorem{lemma}[theorem]{Lemma}

\newtheorem{corollary}[theorem]{Corollary}

\theoremstyle{remark}\newtheorem{remark}[theorem]{Remark}
\theoremstyle{remark}\newtheorem{example}{Example}[section]

\newcounter{ticklistc}

\newcommand{\Z}{\mathbb Z}
\newcommand{\C}{\mathbb C}
\newcommand{\R}{\mathbb R}

\newcommand{\K}{\mathcal K}

\newcommand{\D}{\mathcal D}

\newcommand{\SI}{\Sigma}

\newcommand{\G}{\Gamma}

\newcommand{\A}{\mathit{Arf}}
\newcommand{\eps}{\epsilon}
\newcommand{\sgn}{\mbox{sgn}}

\newcommand{\e}{\varepsilon}

\newcommand{\Pf}{\mbox{Pf}}

\begin{document}

\title{The geometry of dimer models}

\author{David Cimasoni}

\begin{abstract}
This is an expanded version of a three-hour minicourse given at the winterschool {\em Winterbraids IV} held in Dijon in February 2014.
The aim of these lectures was to present some aspects of the dimer model to a geometrically minded audience.
We spoke neither of braids nor of knots, but tried to show how several geometrical tools that we know and love (e.g. (co)homology, spin structures, real
algebraic curves) can be applied to very natural problems in combinatorics and statistical physics.
These lecture notes do not contain any new results, but give a (relatively original) account of the works of Kasteleyn~\cite{Ka3},
Cimasoni-Reshetikhin~\cite{C-RI} and Kenyon-Okounkov-Sheffield~\cite{KOS}.

\end{abstract}

\maketitle

\tableofcontents


\section*{Foreword}

These lecture notes were originally not intended to be published, and the lectures were definitely not prepared with this aim in mind.
In particular, I would like to stress the fact that they do not contain any new results, but only an exposition of well-known results in the field.
Also, I do not claim this treatement of {\em the geometry of dimer models\/} to be complete in any way.
The reader should rather take these notes as a personal account by the author of some selected chapters where the words {\em geometry\/} and {\em dimer models\/}
are not completely irrelevant, chapters chosen and organized in order for the resulting story to be almost self-contained, to have a natural beginning, and a happy ending.


\section{Introduction}

Let $\Gamma$ be a finite connected graph, with vertex set $V(\G)$ and edge set $E(\G)$. A {\em perfect matching\/} on $\G$ is a choice of edges of $\G$ such that each vertex of
$\G$ is adjacent to exactly one of these edges. In statistical mechanics, a perfect matching on $\G$ is also known as a {\em dimer configuration\/} on $\G$, and
the edges of the perfect matching are called {\em dimers\/}. We shall denote by $\D(\G)$ the set of dimer configurations on $\G$. An example
is given in Figure~\ref{ex-dimer}.

The first natural question to address is whether a given finite graph $\Gamma$ admits a perfect matching at all.
Clearly, $\G$ needs to have an even number of vertices, but this condition is not sufficient: the $\Ydown$-shaped graph has 4 vertices but no perfect matching.
There is an efficient algorithm to answer this question: it is the so-called {\em Hungarian method\/} in the case of bipartite graphs~\cite{L-P},
that was extended by Edmonds~\cite{Edm} to the general case.
We refer to these sources for purely combinatorial aspects of matching theory.
Unless otherwise stated, we will only
consider finite graphs which admit perfect matchings. In
particular, all the graphs will have an even number of vertices.

\begin{figure}[Htb]
\centerline{\psfig{figure=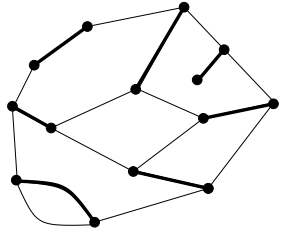,height=3cm}}
\caption{A dimer configuration on a graph.}
\label{ex-dimer}
\end{figure}

The second question is: given a fixed finite graph $\G$, can one count the number $\#\D(\G)$ of perfect matchings on $\G$~? The answer to this question turns out to be
much more subtle. As proved by Valient~\cite{Val}, there is no efficient algorithm to compute this number for an arbitrary finite graph.
However, in the case of planar graphs (and more generally, of graphs of a fixed genus), there is a polynomial time algorithm to compute the number of perfect matchings.
This will be the subject of a good part of these notes.

In fact, it is possible to compute not only the number of dimer configurations, but also the weighted sum of these configurations, that now we define.
An {\em edge-weight system\/} on $\G$ is a positive real-valued function $\nu$ on the set of edges of $\G$. Such a weight system $\nu$ defines a probability measure on the
set of dimer configurations:
\[
\mathbb{P}(D)=\frac{\nu(D)}{Z(\G;\nu)},
\]
where $\nu(D)=\prod_{e\in D}\nu(e)$ and
\[
Z(\G;\nu)=\sum_{D\in\D(\G)}\nu(D).
\]
The normalization constant $Z(\G;\nu)$ is called the {\em partition function\/} for the {\em dimer model\/} on the graph $\G$ with weight system $\nu$.
Note that if the weight system is constant equal to $1$, then $Z(\G;\nu)$ is nothing but the number $\#\D(\G)$ of dimer configurations on $\G$.


\section{Dimers and Pfaffians}

Recall that the determinant of a skew-symmetric matrix $A=(a_{ij})$ of size $2n$ is the square of a polynomial in the $a_{ij}$'s.
This square root, called the {\em Pfaffian\/} of $A$, is given by
\[
\Pf(A)=\frac{1}{2^n n!}\sum_{\sigma\in S_{2n}}\sgn(\sigma) a_{\sigma(1)\sigma(2)}\cdots a_{\sigma(2n-1)\sigma(2n)},
\]
where the sum is over all permutations of $\{1,\dots,2n\}$ and $\sgn(\sigma)\in\{\pm 1\}$ denotes the signature of $\sigma$. Since $A$ is skew-symmetric, the monomial corresponding to
$\sigma\in S_{2n}$ only depends on the matching of $\{1,\dots,2n\}$ into unordered pairs $\{\{\sigma(1),\sigma(2)\},\dots,\{\sigma(2n-1),\sigma(2n)\}\}$.
As there are $2^n n!$ different permutations defining the same matching, we get
\[
\Pf(A)=\sum_{[\sigma]\in\Pi}\sgn(\sigma) a_{\sigma(1)\sigma(2)}\cdots a_{\sigma(2n-1)\sigma(2n)},
\]
where the sum is on the set $\Pi$ of matchings of $\{1,\dots,2n\}$. Note that there exists a skew-symmetric version of the Gauss elimination algorithm, based on the equality
$\Pf(BAB^T)=\det(B)\Pf(A)$. This allows us to compute the Pfaffian of a matrix of size $2n$ in $\mathcal{O}(n^3)$ time.

\medskip

The celebrated {\em FKT algorithm\/} -- named after Fisher, Kasteleyn and Temperley~\cite{Fi1,Ka1,F-T} --
allows us to express the dimer partition function of a certain class of graphs as a Pfaffian. It is based on the following beautifully simple computation.

Enumerate the vertices of $\G$
by $1,2,\dots,2n$, and fix an arbitrary orientation $K$ of the edges of $\G$. Let $A^K=(a_{ij}^K)$ denote the associated
weighted skew-adjacency matrix; this is the $2n\times 2n$ skew-symmetric matrix whose coefficients
are given by
\[
a^K_{ij}=\sum_{e}\e_{ij}^K(e)\nu(e),
\]
where the sum is on all edges $e$ in $\Gamma$ between the vertices $i$ and $j$, and
\[
\e^K_{ij}(e)=
\begin{cases}
\phantom{-}1 & \text{if $e$ is oriented by $K$ from $i$ to $j$;} \\
-1 & \text{otherwise.}
\end{cases}
\]
An example is given in Figure~\ref{fig:baby1}.

\begin{figure}[Htb]
\labellist\small\hair 2.5pt
\pinlabel {$1$} at 5 100
\pinlabel {$2$} at 113 100
\pinlabel {$3$} at 230 152 
\pinlabel {$4$} at 230 5
\pinlabel {$\nu_1$} at 58 120
\pinlabel {$\nu_2$} at 58 33
\pinlabel {$\nu_3$} at 150 128
\pinlabel {$\nu_5$} at 150 26
\pinlabel {$\nu_4$} at 195 80
\pinlabel {$\begin{pmatrix}0&\nu_1+\nu_2&0&0\cr -\nu_1-\nu_2&0&-\nu_3&\nu_5\cr 0&\nu_3&0&\nu_4\cr 0&-\nu_5&-\nu_4&0\end{pmatrix}$} at 390 80
\endlabellist
\centerline{\psfig{file=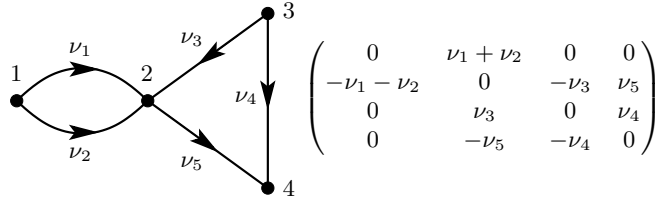,height=2.5cm}}
\caption{An oriented weighted graph, and the corresponding weighted skew-adjacency matrix.}
\label{fig:baby1}
\end{figure}

Now, let us consider the Pfaffian of this matrix. A matching of $\{1,\dots,2n\}$ contributes to $\Pf(A^K)$ if and only if it is realized by a dimer configuration
on $\G$, and this contribution is $\pm \nu(D)$. More precisely,
\begin{equation}
\label{equ:Pf}
\Pf(A^K)=\sum_{D\in\D(\G)}\e^K(D)\nu(D),
\end{equation}
where the sign $\e^K(D)$ can be computed as follows: if the dimer configuration $D$ is given by edges $e_1,\dots,e_n$ matching vertices $i_\ell$ and $j_\ell$ for $\ell=1,\dots,n$,
let $\sigma$ denote the permutation sending
$(1,\dots, 2n)$ to $(i_1,j_1,\dots,i_n,j_n)$; the sign is given by
\[
\e^K(D)=\sgn(\sigma)\prod_{\ell=1}^n\e^K_{i_\ell j_\ell}(e_\ell).
\]
The problem of expressing $Z(\G;\nu)$ as a Pfaffian now boils down to finding an orientation $K$ of the edges of $\G$
such that $\e^K(D)$ does not depend on $D$.

Let us fix $D,D'\in\D(\G)$, and try to compute the product $\e^K(D)\e^K(D')$. Note that the symmetric difference $D\Delta D'$ consists of a
disjoint union of simple closed curves $C_1,\dots,C_m$ of even length in $\G$.
Since the matchings $D$ and $D'$ alternate along these cycles, one can choose permutations $\sigma$ (resp. $\sigma'$) representing $D$ (resp. $D'$) such that $\tau:=\sigma'\circ\sigma^{-1}\in S_{2n}$
is the product of the cyclic permutations defined by the cycles $C_1,\dots,C_m$. Using this particular choice of representatives, and the fact that $\sgn(\tau)=(-1)^m$, we find
\begin{equation}
\label{equ:cc}
\e^K(D)\e^K(D')=\prod_{i=1}^m(-1)^{n^K(C_i)+1},
\end{equation}
where $n^K(C_i)$ denotes the number of edges of $C_i$ where a fixed orientation of $C_i$ differs from $K$.
Since $C_i$ has even length, the parity of this number is independent of the chosen orientation of $C_i$.

Therefore, we are now left with the problem of finding an orientation $K$ of $\G$ with the following propery: for any cycle $C$ of even
length such that $\G\setminus C$ admits a dimer configuration, $n^K(C)$ is odd. Such an orientation was called
{\em admissible\/} by Kasteleyn~\cite{Ka3}; nowadays, the term of {\em Pfaffian orientation\/} is commonly used.
By the discussion above, if $K$ is a Pfaffian orientation, then $Z(\G;\nu)=|\Pf(A^K)|=\sqrt{|\det(A^K)|}$.


\section{Kasteleyn's theorem}

Kasteleyn's celebrated theorem asserts that every planar graph admits a Pfaffian orientation. More precisely, let
$\G$ be a graph embedded in the plane. Each face $f$ of $\G\subset\R^2$ inherits the (say, counterclockwise)
orientation of $\R^2$, so $\partial f$ can be oriented as the boundary of the oriented face $f$.

\begin{theorem}[Kasteleyn~\cite{Ka2,Ka3}]
\label{thm:Kast}
Given $\G\subset\R^2$, there exists an orientation $K$ of $\G$ such that, for each face $f$
of $\G\subset\R^2$, $n^K(\partial f)$ is odd. Furthermore, such an orientation is Pfaffian, so $Z(\G;\nu)=|\Pf(A^K)|$.
\end{theorem}

A striking consequence of this theorem is that it enables to compute the dimer partition function for planar graphs in polynomial time.

Following Kasteleyn, we shall break the proof into three lemmas. Also, we shall say that an orientation $K$ satisfies the {\em Kasteleyn condition\/}
around a face $f$ if $n^K(\partial f)$ is odd. An orientation satisfying the Kasteleyn condition around each face is a {\em Kasteleyn orientation\/}, and the corresponding
skew-adjacency matrix a {\em Kasteleyn matrix\/}.

\begin{lemma}
\label{lemma:A}
There exists a Kasteleyn orientation on any planar graph $\G\subset\R^2$.
\end{lemma}
\begin{proof}
Fix a spanning tree $T$ in $\G^*\subset S^2$, the graph dual to the graph $\G\subset S^2$, rooted at the vertex corresponding to the outer face of $\G\subset S^2$.
Orient arbitrarily the edges of $\G$ which do not intersect $T$. Moving down the tree from the leaves to the root, remove the edges of $T$ one by one, and orient the corresponding
edge of $\G$ in the unique way such that the Kasteleyn condition is satisfied around the newly oriented face.
When all the edges of $T$ are removed, we are left with an orientation $K$ of $\G$ satisfying the Kasteleyn condition around each face of $\G\subset\R^2$.
This is illustrated in Figure~\ref{fig:alg}.
\end{proof}

\begin{figure}[Htb]
\centerline{\psfig{file=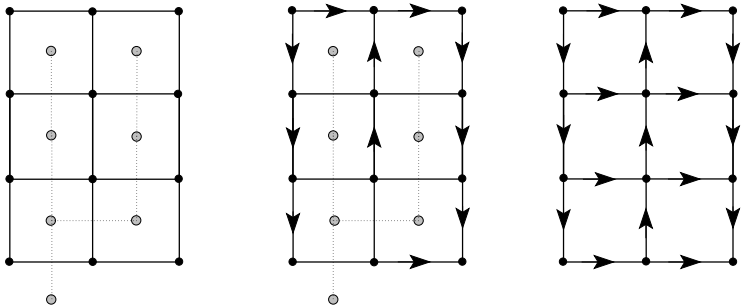,height=4cm}}
\caption{Explicit construction of a Kasteleyn orientation.}
\label{fig:alg}
\end{figure}

One might wonder what happens with the outer face of $\G\subset\R^2$. It turns out that the orientation $K$ of $\G\subset\R^2$ also satisfies the Kasteleyn condition around the outer
face if and only if $\G$ has an even number of vertices. This will follow from a more general statement in Section~\ref{sec:part-general}.

\begin{lemma}
\label{lemma:B}
Let $K$ be a Kasteleyn orientation on $\G\subset\R^2$. Then, given any counterclockwise oriented simple closed curve $C$ in $\G$,
$n^K(C)$ and the number of vertices of $\G$ enclosed by $C$ have opposite parity.
\end{lemma}
\begin{proof}
Let us check this result by induction on the number $k\ge 1$ of faces enclosed by $C$. The case $k=1$ is exactly the Kasteleyn condition, so let us assume the statement true for any
simple closed curve enclosing up to $k\ge 1$ faces, and let $C=\partial S$ be a simple closed curve enclosing $k+1$ faces and $m$ vertices.
Pick a face $f$ in $S$ such that the closure of $f$ intersects $C$ in an interval. (This always exists, except in the case where a unique face $f$ meets $C$ in more than a vertex; this degenerate case
can be treated in a similar way.)
Then, $S':=S\setminus f$ contains $k$ faces, and $C':=\partial S'$ is still a simple closed curve. Also, $f$ meets $S'$ in an interval made of a certain number -- say $\ell\ge 1$ -- of edges of $\G$.
It follows that $S'$ contains $m-\ell+1$ vertices. Using the induction hypothesis and the Kasteleyn condition around $f$, we get the modulo 2 equality
\[
n^K(C)=n^K(C')+n^K(\partial f)-\ell\equiv(m-\ell)+1-\ell\equiv m+1,
\]
and the lemma is proved.
\end{proof}

\begin{lemma}
\label{lemma:C}
Any Kasteleyn orientation on $\G\subset\R^2$ is a Pfaffian orientation.
\end{lemma}
\begin{proof}
Let $C$ be a cycle of even length in $\G$ such that $\G\setminus C$ admits a dimer configuration $D$. This implies that $C$ encloses an even number of vertices, as these vertices are matched 
by $D$. By Lemma~\ref{lemma:B}, $n^K(C)$ is odd, so $K$ is a Pfaffian orientation.
\end{proof}

\begin{figure}[Htb]
\centerline{\psfig{file=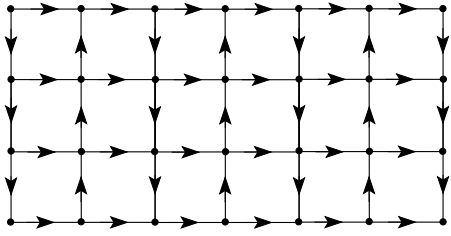,height=3cm}}
\caption{A Kasteleyn orientation on a square lattice.}
\label{fig:square}
\end{figure}

\begin{example}
\label{ex:baby1}
The orientation given in Figure~\ref{fig:baby1} is Kasteleyn. The corresponding Pfaffian is equal to
\[
\Pf(A^K)=\Pf\begin{pmatrix}0&\nu_1+\nu_2&0&0\cr -\nu_1-\nu_2&0&-\nu_3&\nu_5\cr 0&\nu_3&0&\nu_4\cr 0&-\nu_5&-\nu_4&0\end{pmatrix}=\nu_1\nu_4+\nu_2\nu_4.
\]
This is clearly equal to the dimer partition function $Z(\G;\nu)$.
\end{example}

As a less trivial example, let us now consider the case of the square lattice. Note that this example is a very special case of Theorem~\ref{thm:free} below.

\begin{example}\label{ex:square}
Let $\G_{mn}$ be the $m\times n$ square lattice, and let $\nu$ be the edge weight system assigning to each horizontal (resp. vertical) edge the weight $x$ (resp. $y$). If neither $m$ nor $n$ are
even, then $\G_{mn}$ has an odd number of vertices and the partition function $Z(\G_{mn};\nu)$ vanishes. Let us therefore assume that $m$ is even. Let $K$ be the orientation
of $\G_{mn}\subset\R^2$ illustrated in Figure~\ref{fig:square}. This orientation is obviously Kasteleyn, so $Z(\G_{mn};\nu)=|\Pf(A^K)|=\sqrt{\det(A^K)}$, where $A_K$ is the
$mn\times mn$ weighted skew-adjacency matrix associated to $K$. The determinant of $A^K$ can be computed explicitly, leading to the following formula (see~\cite{Ka1}):
\[
Z(\G_{mn};x,y)=\prod_{k=1}^{\frac{1}{2}m}\prod_{\ell=1}^{n}2\Big(x^2\cos^2\frac{k\pi}{m+1}+y^2\cos^2\frac{\ell\pi}{n+1}\Big)^{1/2}.
\]
Using this equation, Kasteleyn computed the following limit:
\[
\lim_{n\to\infty}\frac{1}{n^2}\log Z(\G_{nn};x,y)=\frac{1}{\pi^2}\int_0^{\frac{\pi}{2}}\!\!\!\int_0^{\frac{\pi}{2}}\log(4(x^2\cos^2\varphi+y^2\cos^2\psi))d\varphi d\psi.
\]
In particular, setting $x=y=1$ in the first equation above, we find that the number of dimer configurations on the $m\times n$ square lattice is equal to
\[
\#\D(\G_{mn})=\prod_{k=1}^{\frac{1}{2}m}\prod_{\ell=1}^{n}2\Big(\cos^2\frac{k\pi}{m+1}+\cos^2\frac{\ell\pi}{n+1}\Big)^{1/2}.
\]
For example, the $4\times 3$ lattice illustrated in Figure~\ref{fig:alg} admits 11 dimer configurations.
For $n\times n$ lattices with $n=2,4,6,8$, this number is equal to $2,36,6728,12988816$, respectively. Finally, setting $x=y=1$ in the second equation displayed above, Kasteleyn
established that this number grows as
\[
\#\D(\G_{nn})\sim e^{\frac{G}{\pi}n^2},
\]
where $G=1^{-2}-3^{-2}+5^{-2}-7^{-2}+\dots=0.915965\dots$ is Catalan's constant.
\end{example}

\begin{remark}
If a finite graph $\G$ is bipartite (with $B$ black and $W$ white vertices), then any associated skew-adjacency matrix will be of the form
\[
A^K=\begin{pmatrix}0&M\cr -M^T&0\end{pmatrix},
\]
with $M$ of size $B\times W$. If $\G$ admits a dimer configuration, then this matrix is square (say, of size $k$), and
\[
\Pf(A^K)=(-1)^{\frac{k(k-1)}{2}}\det(M).
\]
Therefore, if $\G$ is a planar graph and $K$ is a Kasteleyn orientation, then
\[
Z(\G;\nu)=|\det(M)|.
\]
\end{remark}

\medskip

What about non-planar graphs? There is no hope to extend Theorem~\ref{thm:Kast} {\em verbatim\/} to the general case, as some graphs
do not admit a Pfaffian orientation. (The complete bipartite graph $K_{3,3}$ is the simplest example).
More generally, enumerating the dimer configurations on a graph is a $\#P$-complete problem \cite{Val}. 

However, Kasteleyn~\cite{Ka1} was able to compute the dimer partition function for the square lattice on the torus using 4 Pfaffians (see Example~\ref{ex:toric} below).
He also stated without further detail that the partition function for a graph of genus $g$ requires $2^{2g}$ Pfaffians~\cite{Ka2}. Such a formula was found much later by
Tesler~\cite{Tes} and Gallucio-Loebl~\cite{G-L}, independently. (See also~\cite{D96}.) These authors generalize by brute force, so to speak, the combinatorial proof
given by Kasteleyn in the planar case and in the case of the biperiodic square lattice.

The aim of the following two sections is to explain an alternative, geometric version of this Pfaffian formula~\cite{C-RI,C-RII,Cim}.
This approach relies on some geometric tools that we very briefly recall in the following (somewhat dry) section.


\section{Homology, quadratic forms and spin structures}

The planarity of $\G$ was used in a crucial way in the proof of Kasteleyn's theorem: we used the fact that a cycle in a planar graph bounds a collection of faces.
In general, any finite graph can be embedded in a closed orientable surface $\Sigma$ of some genus $g$, but if $g$ is non-zero, then some cycles might not bound faces.

\medskip

There is a standard tool in algebraic topology to measure how badly the fact above does not hold: it is called {\em homology\/}, and we now briefly recall its definition in our
context.

Given a graph $\G\subset\SI$ whose complement consists of topological discs, let $C_0$ (resp. $C_1$, $C_2$) denote the $\Z_2$-vector space with basis the set of vertices
(resp. edges, faces) of $\G\subset\SI$. Also, let $\partial_2\colon C_2\to C_1$ and $\partial_1\colon C_1\to C_0$
denote the {\em boundary operators\/} defined in the obvious way. Since $\partial_1\circ\partial_2$ vanishes, the space of {\em $1$-cycles\/} $\mathrm{ker}(\partial_1)$
contains the space $\partial_2(C_2)$ of {\em $1$-boundaries\/}. The {\em first homology space\/} $H_1(\SI;\Z_2):=\mathrm{ker}(\partial_1)/\partial_2(C_2)$ turns out not to depend
on $\G$, but only on $\SI$: it has dimension $2g$, where $g$ is the genus of the closed connected orientable surface $\SI$.
Note that the intersection of curves defines a non-degenerate bilinear form on $H_1(\SI;\Z_2)$, that will be denoted by $(\alpha,\beta)\mapsto \alpha\cdot\beta$.

\medskip

We now turn to quadratic forms. Let $H$ be a $\Z_2$-vector space endowed with a non-degenerate bilinear form~$(\alpha,\beta)\mapsto \alpha\cdot\beta$.
A {\em quadratic form\/} on $(H,\cdot)$ is a map
$q\colon H\to\Z_2$ such that $q(\alpha+\beta)=q(\alpha)+q(\beta)+\alpha\cdot\beta$ for all $\alpha,\beta\in H$. Note that there are exactly $|H|$ quadratic forms on $(H,\cdot)$,
as the set of such forms is an affine space over $\mathit{Hom}(H,\Z_2)$. Furthermore, it can be showed~\cite{Arf} that these forms are classified by their {\em Arf invariant\/}
$\A(q)\in\Z_2$.
We shall need a single property of this invariant, namely that it satisfies the equality
\begin{equation}
\label{equ:Arf}
\frac{1}{\sqrt{|H|}}\sum_{q}(-1)^{\A(q)+q(\alpha)}=1
\end{equation}
for any $\alpha\in H$, where the sum is over all quadratic forms on $(H,\cdot)$.

\medskip

Quadratic forms are the algebraic avatar of a geometric object called a {\em spin structure\/}. We shall not go into the trouble of giving the formal definition (see e.g.~\cite{Ati});
let us simply recall that any spin structure on a closed orientable surface is given by a vector field with zeroes of even index.

The relationship between spin structures and quadratic forms is given by the following classical result of Johnson~\cite{Joh}.
Consider a spin structure $\lambda$ on $\SI$ represented by a vector field $Y$ on $\SI$ with zeroes of even index. Given a piecewise smooth closed curve
$\gamma$ in $\SI$ avoiding the zeroes of $Y$, let $\mathit{rot}_\lambda(\gamma)\in 2\pi\Z$ denote the rotation angle of the velocity vector of $\gamma$ with respect to $Y$.
Then, given a homology class $\alpha\in H_1(\SI;\Z_2)$ represented by the disjoint union of oriented simple closed curves $\gamma_j$, the equality
$(-1)^{q_\lambda(\alpha)}=\prod_j-\exp\left(\textstyle{\frac{i}{2}}\mathit{rot}_\lambda(\gamma_j)\right)$
gives a well-defined quadratic form on $(H_1(\SI;\Z_2),\cdot)$. Furthermore, Johnson's theorem asserts that the mapping~$\lambda\mapsto q_\lambda$
defines a bijection between the set of spin structures on~$\SI$ and the set of quadratic forms on~$(H_1(\SI;\Z_2),\cdot)$.


\section{The partition function for general graphs}
\label{sec:part-general}

Let $\G\subset\SI$ be a fixed {\em surface graph\/}, i.e. a graph embedded in an orientable closed surface $\SI$ such that $\SI\setminus\G$ consists of topological disks.
Throughout this section, we shall use the same notation $X$ for the surface graph and the induced cellular decomposition of $\SI$.

We will now try to encode combinatorially a spin structure on a surface $\SI$, or
equivalently, a vector field on $\SI$ with isolated zeroes of even index. This will lead us to a very natural ``re-discovery" of the notion of a Kasteleyn orientation. 
Let $X$ be a fixed cellular decomposition of $\SI$.

$\bullet$
To construct a (unit length) vector field along the 0-skeleton $X^0$, we just need to specify one tangent
direction at each vertex of $X$. Such an information is given by a dimer configuration $D$ on $X^1$: at each vertex,
point in the direction of the adjacent dimer.

$\bullet$
This vector field along $X^0$ extends to a unit vector field on $X^1$, but not uniquely. Roughly speaking,
it extends in two different natural ways along each edge of $X^1$, depending on the sense of rotation of the resulting
vector field. We shall
encode this choice by an orientation $K$ of the edges of $X^1$, together with the following convention: moving along
an oriented edge, the tangent vector first rotates counterclockwise until it points in the direction of the edge,
then rotates clockwise until it points backwards, and finally rotates counterclockwise until it coincides with
the tangent vector at the end vertex. This is illustrated in Figure~\ref{fig:vector}.

\begin{figure}[Htb]
\labellist\small\hair 2.5pt
\pinlabel {$K$} at 220 165
\pinlabel {$D$} at 75 60
\pinlabel {$D$} at 395 190
\endlabellist
\centerline{\psfig{file=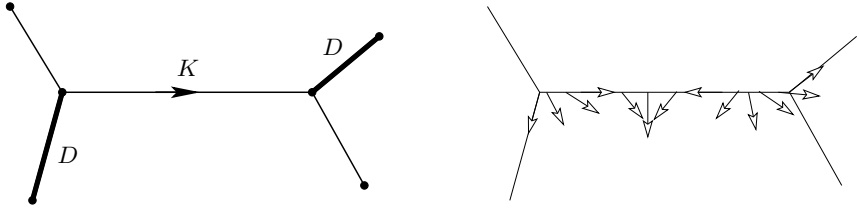,height=2.7cm}}
\caption{Construction of the vector field along the 1-skeleton of $X$.}
\label{fig:vector}
\end{figure}

$\bullet$
Each face of $X$ being homeomorphic to a 2-disc, the unit vector field defined along $X^1$ naturally extends to
a vector field $Y_D^K$ on $X$, with one isolated zero in the interior of each face. 

One easily checks that for each face $f$ of $X$, the index of the zero of $Y_D^K$ in $f$ has the parity of $n^K(\partial f)+1$.
Therefore, the vector field $Y_D^K$ defines a spin structure if and only if $K$ satisfies the Kasteleyn condition around each face.
Hence, we shall say that an orientation $K$ of the 1-cells of a surface graph $X$ is a {\em Kasteleyn orientation on $X$\/} if,
for each face $f$ of $X$, the number $n^K(\partial f)$ is odd.

Given a Kasteleyn orientation on $X$, there is an obvious way to obtain another one: pick a vertex of $X$ and flip
the orientation of all the edges adjacent to it. Two Kasteleyn orientations are said to be {\em equivalent\/} if
they can be related by such moves. Let us denote by $\K(X)$ the set of equivalence classes of Kasteleyn
orientations on $X$. 

Using standard cohomological techniques (see~\cite{C-RI}), it can be showed that a surface graph $X$ admits a Kasteleyn orientation if and only if $X$ has an even number of vertices.
In this case, the set $\K(X)$ admits a freely transitive action of the group $H^1(\SI;\Z_2)=\mathit{Hom}(H_1(\SI;\Z_2),\Z_2)$.

\medskip

Summarizing this section so far, a dimer configuration $D$ on $X^1$ and a Kasteleyn orientation $K$ on $X$ determine a spin structure on $\SI$,
i.e: a quadratic form $q^K_D\colon H_1(\SI;\Z_2)\to\Z_2$. Furthermore, this quadratic form can be computed explicitely. If $C$ is an oriented simple closed curve in $\G$, then
\begin{equation}
\label{eqn:q}
q^K_D([C])=n^K(C)+\ell_{D}(C)+1,
\end{equation}
where $\ell_{D}(C)$ denotes the number of vertices of $C$ where the adjacent dimer of $D$ sticks out to the left of $C$.
With this formula, one immediately sees that the quadratic form $q^K_D$ only depends on the equivalence class $[K]$ of the Kasteleyn orientation $K$. With some more work, one can
show that the map $[K]\mapsto q^K_D$ defines an $H^1(\SI;\Z_2)$-equivariant bijection from the set of equivalence classes of Kasteleyn orientations on $X$ to the set of
spin structures on $\SI$. Therefore, Equation~(\ref{equ:Arf}) translates into
\begin{equation}
\label{eqn:Arf}
\frac{1}{2^g}\sum_{[K]\in\K(X)}(-1)^{\A(q^K_{D})+q^K_{D}(\alpha)}=1
\end{equation}
for all $\alpha\in H_1(\SI;\Z_2)$.

However, the most important outcome of this discussion is that we now know ``for free'' that the formula displayed above gives a well-defined quadratic form on $H_1(\SI;\Z_2)$.
Let us use this fact to solve our original problem, i.e. the computation of the dimer partition function.

\medskip

Let $\G$ be a finite connected graph endowed with an edge weight system $\nu$. If $\G$ does not admit any dimer
configuration, then the partition function $Z(\G;\nu)$ is obviously zero, so let us assume that $\G$ admits a dimer
configuration $D_0$. Enumerate the vertices of $\G$ by $1,2,\dots,2n$ and embed $\G$ in a closed orientable surface
$\SI$ of genus $g$ as the 1-skeleton of a cellular decomposition $X$ of $\SI$.

Since $\G$ has an even number of vertices, it admits a Kasteleyn orientation $K$. Replacing it with an equivalent orientation, we can assume that $\e^K(D_0)=1$.
The Pfaffian of the associated Kasteleyn matrix then satisfies
\begin{align*}
\Pf(A^K)&\overset{(\ref{equ:Pf})}{=}\sum_{D\in\D(\G)}\e^K(D_0)\e^K(D)\,\nu(D)\\
	&\overset{(\ref{equ:cc})}{=}\sum_{D\in\D(\G)}(-1)^{\sum_i (n^K(C_i)+1)}\nu(D),
\end{align*}
where $D\Delta D_0=\bigsqcup_iC_i$. Note that given any vertex of $C_i$, the adjacent dimer of $D_0$ lies on $C_i$, so that $\ell_{D_0}(C_i)=0$.
Since the cycles $C_i$ are disjoint, the quadratic form is linear on them and we get
\[
\sum_i(n^K(C_i)+1)=\sum_i(n^K(C_i)+\ell_{D_0}(C_i)+1)\overset{(\ref{eqn:q})}{=}\sum_iq^K_{D_0}(C_i)=q^K_{D_0}([D\Delta D_0]).
\]
Therefore, for every element $[K]$ of $\K(X)$, we have
\begin{equation}
\label{eqn:Pf}
\Pf(A^K)=\sum_{D\in\D(\G)}(-1)^{q^K_{D_0}([D\Delta D_0])}\nu(D)=\sum_{\alpha\in H_1(\SI;\Z_2)}(-1)^{q^K_{D_0}(\alpha)}\sum_{[D\Delta D_0]=\alpha}\nu(D),
\end{equation}
the last sum being over all $D\in\D(\G)$ such that the homology class of $D\Delta D_0$ is $\alpha$. This leads to:
\begin{align*}
Z(\G;\nu)&=\sum_{\alpha\in H_1(\SI;\Z_2)}\sum_{[D\Delta D_0]=\alpha}\nu(D)\\
&\overset{(\ref{eqn:Arf})}{=}\sum_{\alpha\in H_1(\SI;\Z_2)}\Big(\frac{1}{2^g}\sum_{[K]\in\K(X)}(-1)^{\A(q^K_{D_0})+q^K_{D_0}(\alpha)}\Big)\sum_{[D\Delta D_0]=\alpha}\nu(D) \\
&=\frac{1}{2^g}\sum_{[K]\in\K(X)}(-1)^{\A(q^K_{D_0})}\sum_{\alpha\in H_1(\SI;\Z_2)}(-1)^{q^K_{D_0}(\alpha)}\sum_{[D\Delta D_0]=\alpha}\nu(D)\\
&\overset{(\ref{eqn:Pf})}{=}\frac{1}{2^g}\sum_{[K]\in\K(X)}(-1)^{\A(q^K_{D_0})}\Pf(A^K)\,.
\end{align*}

We have proved the following {\em Pfaffian formula\/}, first obtained in~\cite{C-RI}.

\begin{theorem}
\label{thm:Pf}
Let $\G$ be a graph embedded in a closed oriented surface $\SI$ of genus $g$ such that $\SI\setminus\G$ consists of topological disks.
Then, the partition function of the dimer model on $\Gamma$ is given by the formula
\[
Z(\G;\nu)=\frac{1}{2^{g}}\sum_{[K]\in\K(\G\subset\SI)}(-1)^{\A(q^{K}_{D_0})}\Pf(A^{K}).
\]
\end{theorem}

Let us conclude this section with one example.

\begin{example}
\label{ex:toric}
Consider the $m\times n$ square lattice wrapped up around the torus.
As usual, let $\nu$ be the edge weight system assigning to each horizontal (resp. vertical) edge the weight $x$ (resp. $y$), and let us assume without loss of generality that $m$ is even.
Let $K$ be the orientation of $\G\subset\SI$ illustrated in Figure~\ref{fig:square}. We get (see~\cite{Ka1}):
\[
Z(\G;x,y)=\frac{1}{2}\left(P_{00}+P_{10}+P_{01}-P_{11}\right),
\]
where
\[
P_{\eps_1\eps_2}=\prod_{k=1}^{\frac{1}{2}m}\prod_{\ell=1}^{n}2\Big(x^2\sin^2\frac{(2\ell+\eps_2-1)\pi}{n}y^2\sin^2\frac{(2k+\eps_1-1)\pi}{m}\Big)^{1/2}.
\]
The Pfaffian $P_{11}$ vanishes (because of the factor corresponding to $k=m/2$ and $\ell=n$), so $Z(\G;x,y)=(P_{00}+P_{10}+P_{01})/2$.
This also shows that, unlike in the planar case, there exist graphs that admit dimer configurations but where some Pfaffian vanishes in the Pfaffian formula.
\end{example}


\section{Special Harnack curves}

We shall now make a small detour into an area {\em a priori\/} totally unrelated: real algebraic geometry. We shall be rather sketchy and refer the interested reader to~\cite{Mik,M-R} for further details.

Consider a real polynomial $P\in\R[z^{\pm 1},w^{\pm 1}]$. The associated {\em real plane curve\/} is
\[
\R A:=\{(z,w)\in(\R^*)^2\,;\,P(z,w)=0\}\subset (\R^*)^2.
\]
One can also consider the associated {\em complex plane curve\/}
\[
A:=\{(z,w)\in(\C^*)^2\,;\,P(z,w)=0\}\subset (\C^*)^2.
\]

It is a classical result due to Harnack~\cite{Har} that the number $c$ of connected components (or {\em ovals\/}) of a real algebraic plane curve of degree $d$ in
$\R P^2\supset(\R^*)^2$ is bounded by
\[
c \le \frac{(d-1)(d-2)}{2}+1.
\]
Which configurations of ovals can be realized is the subject of Hilbert's sixteenth problem, only solved up to degree $d=7$.

More generally, if $P(z,w)=\sum_{(m,n)\in\Z^2}a_{mn}z^mw^n$, one can define:
\begin{itemize}
\item{$\Delta:=\mathit{conv}\{(m,n)\in\Z^2\,|\,a_{mn}\neq 0\}$, the {\em Newton polygon\/} of $P$, with $g$ interior integral points and $s$ sides $\delta_1,\dots,\delta_s$;}
\item{using $\Delta$, a natural compactification $\R T_\Delta$ of $(\R^*)^2$, the associated {\em toric surface\/}, with $\R T_\Delta\setminus(\R^*)^2$ consisting of
$s$ real lines $\ell_1\cup\dots\cup\ell_s$ intersecting cyclically;}
\item{the closure $\R\overline{A}$ of $\R A$ in $\R T_\Delta$.}
\end{itemize}

It turns out that the number of ovals of $\R\overline{A}$ is bounded above by $g+1$, and that $\R\overline{A}$ meets the axis $\ell_i$ corresponding to the side $\delta_i$
at most $d_i$ times, where $d_i-1$ denotes the number of integral points in the interior of $\delta_i$.
We shall say that $\R A$ is a {\em special Harnack curve\/} if the number of ovals of $\R\overline{A}$ is $g+1$, only one of these ovals intersects the axes
$\ell_1\cup\dots\cup\ell_s$, and it does so in the maximal and cyclic way: $d_1$ times $\ell_1$, then $d_2$ times $\ell_2$, and so on.

Note that if $\Delta$ is a triangle, then $\R T_\Delta=\R P^2=(\R^*)^2\cup\ell_1\cup\ell_2\cup\ell_3$, where the $\ell_i$'s are the coordinate axes.
In this case, $\R\overline{A}$ is simply given by the zeroes of the homogeneous polynomial associated to $P$, and the inequality $c\le g+1$ corresponds to Harnack's upper bound.

\begin{example}
\label{ex:Harnack}
Consider the very simple example given by $P(z,w)=a+bz+cw$, with $a,b,c\in\R^*$. The associated Newton polygon is the triangle with
vertices $(0,0)$, $(1,0)$ and $(0,1)$, so we get $g=0$ and $d_1=d_2=d_3=1$. The corresponding curve is a special Harnack curve, since
\[
\R\overline{A}=\{[z:w:t]\in\R P^2\,;\,at+bz+cw=0\}
\]
has one oval and meets each axis once.
\end{example}

The following result of Mikhalkin is of great importance for the application of Harnack curves to the study of dimer models.

\begin{theorem}[Mikhalkin~\cite{Mik,M-R}]
\label{thm:Mik}
Fix $P\in\R[z^{\pm 1},w^{\pm 1}]$ with Newton polygon of positive area. Then, the following are equivalent.
\begin{enumerate}[1.]
\item{$\R A$ is a (possibly singular) special Harnack curve.}
\item{$A\subset(\C^*)^2$ meets every torus $\{(z,w)\in(\C^*)^2\,;\,|z|=r_1\,|w|=r_2 \}$ at most twice.}
\end{enumerate}
\end{theorem}


\section{Bipartite graphs on the torus}
\label{sec:KOS}

After this small detour, let us come back to dimers. Consider the case of a bipartite graph $\G$ (with same number of black and white vertices) embedded in the torus $\mathbb{T}^2$,
and fix a Kasteleyn orientation on $\G\subset\mathbb{T}^2$.

Recall that the corresponding Kasteleyn matrix is of the form $A^K=\left(\begin{smallmatrix}0&M\cr -M^T&0\end{smallmatrix}\right)$,
so $\Pf(A^K)=\pm\det(M)$. Let us modify $M$ as follows: pick oriented simple closed curves $\gamma_x$ and $\gamma_y$ on the torus, transverse to $\G$, representing a basis
of the homology $H_1(\mathbb{T}^2)$, and multiply the coefficient corresponding to each edge $e$ by $z^{\gamma_x\cdot e}w^{\gamma_y\cdot e}$, where $\cdot$ denotes
the intersection number and $e$ is oriented from the white to the black vertex. Let $M(z,w)$ denote the resulting matrix. Its determinant
\[
P(z,w):=\det(M(z,w))\in\R[z^{\pm 1},w^{\pm 1}]
\]
is called the {\em characteristic polynomial\/} of $\G$. We shall say that a weighted bipartite toric graph $\G\subset\mathbb{T}^2$ is {\em non-degenerate\/} if the Newton polygon
of its characteristic polynomial has positive area.

Note that the formula for the dimer partition function now reads
\begin{equation}\label{eqn:Z}
Z(\G;\nu)=\frac{1}{2}\left(-P(1,1)+P(-1,1)+P(1,-1)+P(-1,-1)\right)
\end{equation}
for a well-chosen Kasteleyn orientation.

\begin{figure}[Htb]
\labellist\small\hair 2.5pt
\pinlabel {$a$} at 135 75
\pinlabel {$b$} at 115 22
\pinlabel {$c$} at 235 85
\pinlabel {$\gamma_x$} at 230 2
\pinlabel {$\gamma_y$} at 76 115
\endlabellist
\centerline{\psfig{figure=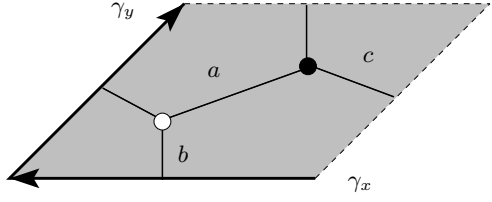,height=2.5cm}}
\caption{A fundamental domain for the hexagonal lattice wrapped up on the torus.}
\label{fig:hexa}
\end{figure}

\begin{example}\label{ex:hexa}
Consider the weighted toric bipartite graph illustrated in Figure~\ref{fig:hexa}. Note that this is nothing but the simplest fundamental domain of the hexagonal lattice.
On this example, the orientation from the white to the black vertex is Kasteleyn, and with the choice of curves $\gamma_x$ and $\gamma_y$ illustrated in this figure,
one computes 
\[
P(z,w)=M(z,w)=a+bz+cw.
\]
As seen in Example~\ref{ex:Harnack}, it is a special Harnack curve.
\end{example}

This is no accident.

\begin{theorem}[Kenyon-Okounkov-Sheffield~\cite{K-O,KOS}]
Let $(\G,\nu)\subset\mathbb{T}^2$ be a non-degenerate weighted bipartite graph embedded in the torus. Then, the real algebraic curve $\R A$ associated with $P(z,w)$
is a (possibly singular) special Harnack curve.
\end{theorem}

By Theorem~\ref{thm:Mik}, we immediately get the following corollary.

\begin{corollary}
\label{cor:max}
For any  non-degenerate weighted bipartite graph $(\G,\nu)\subset\mathbb{T}^2$, the associated complex algebraic curve meets the unit torus
\[
S^1\times S^1=\{(z,w)\in(\C^*)^2\,;\,|z|=|w|=1\}
\]
at most twice.\qed
\end{corollary}

This seemingly technical result is extremely deep, and actually provides the necessary information for the rigorous study of large scale properties of dimer models on biperiodic bipartite graphs.
The paper~\cite{KOS} contains several applications of this result; we shall only explain the simplest of these, namely the computation of the {\em free energy\/}, that we now define.

For a graph $\G$ embedded in the torus, there is a very natural way to make the graph finer and finer: simply consider a fundamental domain for $\G\subset\mathbb{T}^2$ and
paste together $n\times n$ copies of these fundamental domains before wrapping up the torus again. (Topologically, this corresponds to considering a covering map of degree $n^2$.)
Let us denote by $\G_{n}\subset\mathbb{T}^2$ the corresponding graph. Then, the {\em free energy (per fundamental domain)\/} of the model is defined by
\[
\log Z:=\lim_{n\to \infty}\frac{1}{n^2}\log Z(\G_{n};\nu).
\]
The existence of this limit follows from standard subadditivity arguments.
Note that in particular, $\log Z$ evaluated at $\nu=1$ measures the exponential growth of the number of dimer configurations on $\G_{n}$.

\begin{theorem}[Kenyon-Okounkov-Sheffield~\cite{KOS}]\label{thm:free}
For any weighted bipartite graph $\G$ on the torus, the free energy is given by
\[
\log Z=\frac{1}{(2\pi i)^2}\int_{S^1\times S^1}\log|P(z,w)|\frac{dz}{z}\frac{dw}{w}.
\]
\end{theorem}

\begin{proof}[Sketch of proof.]
With Equation~(\ref{eqn:Z}) and Corollary~\ref{cor:max} in hand, the demonstration of Theorem~\ref{thm:free} is quite straightforward. It can be divided into three steps.

\noindent{\em Step 1.} Let us write $P_{n}$ for the characteristic polynomial of $\G_{n}$ and use the notation $Z_{n}^{\theta\tau}:=P_{n}((-1)^\theta,(-1)^\tau)$. We have
\[
|Z_{n}^{\theta\tau}|\overset{(\ref{equ:Pf})}{\le} Z(\G_{n};\nu)\overset{(\ref{eqn:Z})}{=}\frac{1}{2}\left(-Z_{n}^{00}+Z_{n}^{10}+Z_{n}^{01}+Z_{n}^{11}\right)\le2\max_{(\theta,\tau)}|Z_{n}^{\theta\tau}|
\]
for all $\theta,\tau\in\{0,1\}$. It follows that
\[
\lim_{n\to \infty}\frac{1}{n^2}\log Z(\G_{n};\nu)=\lim_{n\to \infty}\frac{1}{n^2}\log\max_{(\theta,\tau)}|Z_{n}^{\theta\tau}|,
\]
that is to say,
\[
\log Z=\lim_{n\to \infty}\frac{1}{n^2}\log\max_{(\theta,\tau)}|P_{n}((-1)^\theta,(-1)^\tau)|.
\]
\noindent {\em Step 2.} Using the symmetry of the twisted Kasteleyn matrix for $\G_{n}$, it can be block-diagonalized into $n^2$ copies of the twisted Kasteleyn matrix for $\G=\G_{1}$. Computing
the determinants leads to the following formula:
\[
P_{n}(z,w)=\prod_{u^n=z}\prod_{v^n=w}P(u,v).
\]
In particular, for $\theta,\tau\in\{0,1\}$, we have
\[
\frac{1}{n^2}\log|P_{n}((-1)^\theta,(-1)^\tau)|=\frac{1}{n^2}\sum_{u^n=(-1)^\theta}\sum_{v^n=(-1)^\tau}\log|P(u,v)|,
\]
which is nothing but a Riemann sum for the integral
\[
\int_0^1\!\!\!\int_0^1\log|P(e^{2\pi i\varphi},e^{2\pi i\psi})|d\varphi d\psi=\frac{1}{(2\pi i)^2}\int_{S^1\times S^1}\log|P(z,w)|\frac{dz}{z}\frac{dw}{w}.
\]
\noindent {\em Step 3.} By step 1, we only need to check that this sum converges for at least one choice of $(\theta,\tau)$. By Corollary~\ref{cor:max}, $P(z,w)$ has at most two (conjugate) zeroes on the unit
torus $S^1\times S^1$. Since the four Riemann sums above are on four staggered lattices, at least three of them converge.
\end{proof}

We conclude these notes with two examples.

\begin{example}
By Example~\ref{ex:hexa} and Theorem~\ref{thm:free}, the free energy for the dimer model on the hexagonal lattice (with fundamental domain as in Figure~\ref{fig:hexa}) is given by
\[
\log Z=\frac{1}{(2\pi)^2}\int_0^{2\pi}\!\!\!\int_0^{2\pi}\log\left|a+be^{i\varphi}+ce^{i\psi}\right|d\varphi d\psi.
\]
\end{example}

\begin{example}
Consider the weighted toric bipartite graph illustrated in Figure~\ref{fig:square2}, which is nothing but the simplest fundamental domain for the {\em bipartite\/} square lattice on the torus.
One easily computes its characteristic polynomial
\[
P(z,w)=y^2(2+(z+z^{-1}))+x^2(2+(w+w^{-1})).
\]
By Theorem~\ref{thm:free}, the corresponding free energy is given by
\begin{align*}
\log Z&=\frac{1}{(2\pi)^2}\int_0^{2\pi}\!\!\!\int_0^{2\pi}\log(2(x^2+y^2+x^2\cos\varphi+y^2\cos\psi))d\varphi d\psi\\
	&=\frac{1}{\pi^2}\int_0^{\pi}\!\!\!\int_0^{\pi}\log(4(x^2\cos^2\varphi+y^2\cos^2\psi))d\varphi d\psi.
\end{align*}
This is four times the value obtained by Kasteleyn  (recall Example~\ref{ex:square}), as it should, the fundamental domain for the bipartite square lattice being of size two by two. 
\end{example}

\begin{figure}[Htb]
\labellist\small\hair 2.5pt
\pinlabel {$x$} at 120 80
\pinlabel {$y$} at 80 120
\pinlabel {$\gamma_x$} at 255 7
\pinlabel {$\gamma_y$} at -20 230
\endlabellist
\centerline{\psfig{figure=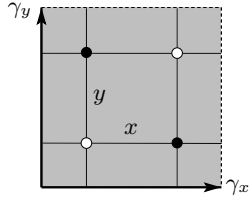,height=2.5cm}}
\caption{A fundamental domain for the bipartite square lattice wrapped up on the torus.}
\label{fig:square2}
\end{figure}

\bibliographystyle{plain}

\bibliography{dimers}

\end{document}